\documentclass{svmult}
\usepackage[utf8]{inputenc}
\usepackage[T1]{fontenc}
\usepackage{newtxmath}
\usepackage{graphicx}
\usepackage{hyperref}

\usepackage[ruled]{algorithm2e}
\DontPrintSemicolon

\def\etal{\textit{et al.}}

\def\Z{\mathbb{Z}}
\def\Fmin{\mathcal{F}_{\mathrm{min}}}

\let\phi=\varphi

\let\meet=\land
\let\join=\lor
\let\bigmeet=\bigwedge
\let\bigjoin=\bigvee

\def\set#1{\{\, #1 \,\}}
\def\abs#1{\mathopen|#1\mathclose|}

\def\class#1{\textsf{#1}}
\def\cpl{\overline}

\def\evol#1{\includegraphics[scale=0.75]{img/#1}}
\newenvironment{evolution}%
{\bgroup%
  \begin{tabular*}{\textwidth}{l@{\extracolsep{\fill}}r}}%
    {\end{tabular*}\egroup}
\def\br{\linebreak[1]} 

\begin{document}
\title{An Invitation to Number-Conserving Cellular Automata}
\author{Markus Redeker}
\institute{Hamburg, Germany. \email{markus2.redeker@mail.de}}
\maketitle

\abstract{Number-conserving cellular automata are discrete dynamical
  systems that simulate interacting particles like e.\,g.\ grains of
  sand. In an earlier paper, I had already derived a uniform
  construction for all transition rules of one-dimensional
  number-conserving automata. Here I show in greater detail how
  one can simulate the automata on a computer and how to find
  interesting rules. I show several rules that I have found this way
  and also some theorems about the space of number-conserving
  automata.}

\section{Introduction}

One can think of cellular automata as systems of particles that may
destroy each other, transform each other or create new particles when
they collide. A special class of cellular automata are therefore those
in which no creation or destruction occurs and the number of particles
stays the same. The particles in such an automaton resemble the
``particles'' of everyday life, like grains of sand or cars in a
street.

In \cite{Redeker2022}, I have derived a uniform method to construct
all one-dimensional number-conserving cellular automata with only one
kind of particles. But I could show only a few of the consequences
that this construction has on the theory of one-dimensional cellular
automata, and also only few interesting example of number-conserving
automata. This is the subject of the present article: It contains
examples for number-conserving automata, helpful hints for those that
want to explore them on their own on a computer, and finally a few
theorems about the space of number-conserving rules.

\begin{figure}[t]
  \begin{evolution}
    \evol{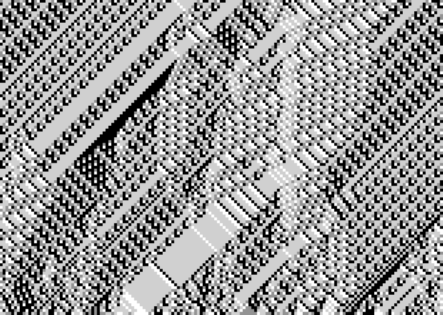} & \evol{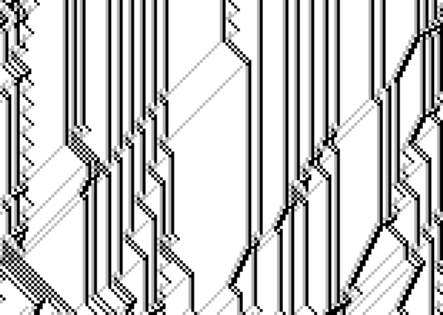} \\
    \evol{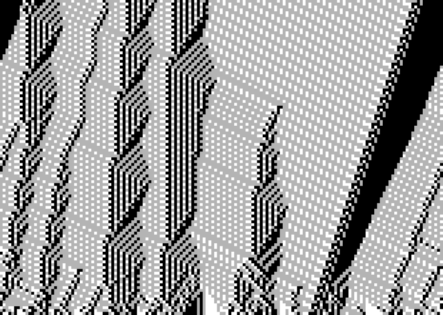} & \evol{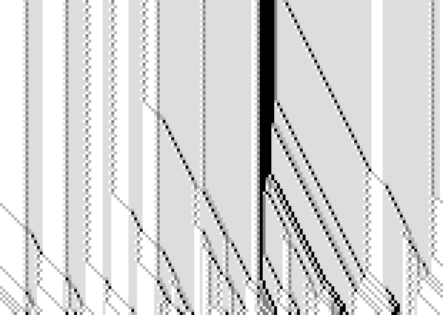}
  \end{evolution}
  \caption{Snapshots of number-conserving cellular automata. \\
    \emph{About the images and parameters:} Time goes upward. White
    cells are empty, black cells totally filled with particles, and
    grey are cells in between. $C$ is the capacity of the cellular
    automaton, $f$ its flow function and
    $\delta = (\delta_0, \dots, \delta_C)$ the vector of initial cell densities: Each
    $\delta_i$ is the fraction of cells in the initial configuration that
    are in state $\pi_i$. \\
    \emph{Top left:} A rule that develops into several complex textures.
    \emph{Parameters:} $C = 3$,
    $f =\br m(12,1;1) \join m(22,1;3) \join m(20,3;-2)$,
    $\delta = (0.25, 0.25,\br 0.25, 0.25)$. \\
    \emph{Top right:} A system of interacting gliders. For this
    pattern, the initial particle density must be
    relatively low. \emph{Parameters:} $C = 2$,
    $f = m(101,0;2) \join m(022,2;2) \join m(012,2;2) \join m(211,1;2)$,
    $\delta = (0.66, 0.17, 0.17)$. \\
    \emph{Bottom left:} A rule where the particles crystallise into regular
    patterns. \emph{Parameters:} $C = 2$,
    $f = m(11,12;0) \join m(02,20;0) \join m(02,11;0) \join m(12,12;1)
    \join m(10,20;1) \join m(20,20;0)$, $\delta = (0.33, 0.33, 0.33)$. \\
    \emph{Bottom right:} In the centre is a block of completely filled
    cells that absorbs the particles right of it. \emph{Parameters:}
    $C = 4$,
    $f = m(3,4;0) \join m(1,2;1) \join m(2,4;0) \join m(3,0;2)$,
    $\delta = (0.96, 0.1, 0.1, 0.1, 0.1)$.}
  \label{fig:snapshots}
\end{figure}
The aim of this all is to show that number-conserving cellular
automata are a research field with many interesting questions and new
phenomena.

\subsection{About number conservation in general}

Number-conserving cellular automata appear in a larger context of
simulating moving particles with cellular automata.

To this context belongs already the very first cellular automaton,
John von Neumann's self-replicator and universal constructor
\cite{Burks1970}: It has ``signals'' moving on ``wires'', but their
number is preserved only when the wires do not branch or end. Later
self-replicators by Codd \cite{Codd1968} and Langton
\cite{Langton1984} also have signals and wires.

For the purpose of particle system simulation, \emph{ad hoc}
constructions are often enough to get number-conserving automata with
the properties one wants. An example are lattice gases
\cite{Wolf-Gladrow2005}, which simulate systems of colliding particles
and are widely used in physics. Another is the Nagel-Schreckenberg
model \cite{NagelSchreckenberg1992} for street traffic. In it, the
``particles'' whose number is conserved are cars. But these models
only provide a subset of all number-conserving cellular automata, and
they do not intend to do anything else.

There are two conceptually different approaches for finding all
number-conserving cellular automata: One can define an invariant,
called ``number of particles'', and then ask for the cellular automata
for which it is preserved. Or one can understand the automaton as
simulating individual particles, and for each time step one needs to
know for each particle to which place it has moved.

For the first approach, one needs to define number conservation
precisely. There are different ways to do that, but Durand, Formenti
and Róka \cite{DurandFormentiEtAl2003} show that some of them are
equivalent. The most influential paper about number conservation via
an invariant is certainly that by Hattori and Takesue
\cite{HattoriTakesue1991} in which they derive a criterion to verify
that a given one-dimensional cellular automaton is number-conserving.
It provides the starting point for most later papers on number
conservation. Boccara, Fuk\'s and Procyk
\cite{Boccara1998,BoccaraFuks2002,FuksProcyk2019} searched with the
criterion for number-conserving cellular automata in one dimension.
Durand \etal\ \cite{DurandEtAl2003} found a version of the criterion
for two dimensions.

The second approach was used by Boccara and Fuk\'s \cite{Boccara1998},
who searched for one-dimensional cellular automata that have an
interpretation in terms of moving particles; Moreira
\cite{Moreira2004} and Ishizaka \etal\ \cite{IshizakaEtAl2015} did the
same in two dimensions. This approach looks at first as if it were the
more restrictive one, but Fuk\'s \cite{Fuks2000} (in one dimension)
and Kari and Taati \cite{KariTaati2008} (in two dimensions) have shown
how to find particle movements for any cellular automata that preserve
the number of particles.

Instead of searching the rule space, one can also try to construct
number-conserving automata directly. In one dimension and for radius
$\frac12$, this was done by Imai and Alhazov \cite{ImaiAlhazov2010},
while Wolnik \etal\ \cite{WolnikEtAl2019} constructed
number-conserving cellular automata for the von Neumann neighbourhood
in any dimension. They all use the invariant viewpoint. The work of
Imai and Alhazov is also a precursor of my own work.

For further references on number conservation see Bhattacharjee
\etal\ \cite[Sec.~4.6]{BhattacharjeeNaskarEtAl2016}, Taati
\cite[Sec.~5]{Taati2012} or the introduction of Wolnik \etal\
\cite{WolnikEtAl2019}.







\section{Summary: How to construct number-conserving cellular
  automata}

In this section I give a short summary of the results in
\cite{Redeker2022}, as far as they are relevant for the present text.

\subsection{Cells and Cell States}
\label{sec:cell-states}

A one-dimensional cellular automaton is a discrete dynamical system in
which at each moment in time, the configuration is a doubly infinite
sequence of small automata, the \emph{cells}.

In every cellular automaton, each cell has one of a finite number of
states and the set of all possible states is the same for all cells.
In a number-conserving automaton, the cells are thought as containers
for a finite number of indistinguishable particles. The state of a
cell is therefore an element of the set
\begin{equation}
  \label{eq:states-set}
  \Sigma = \Sigma_C = \{\pi_0, \dots \pi_C \}
\end{equation}
where $C$ is the \emph{capacity} of the cells, and a cell is in the
state $\pi_i$ if it contains $i$ particles.\footnote{The theory of
  \cite{Redeker2022} also supports cellular automata in which the cell
  state contains more information than just the number of particles,
  but we will not consider them here.}

A configuration is then an element of the set $\Sigma^\Z$. If $a$ is a
configuration, the state of the cell at position $i$ is $a_i$, so that
we can write,
\begin{equation}
  \label{eq:configuration}
  a = \dots a_{-3} a_{-2} a_{-1} a_0 a_1 a_2 a_3 \dots\,.
\end{equation}

We will also need finite sequences of cell states. They are called
\emph{strings} or \emph{neighbourhoods}. The set of all strings is
$\Sigma^* = \bigcup_{k \geq 0} \Sigma^k$. If $b \in \Sigma^k$ is a
string, it is typically written as
\begin{equation}
  \label{eq:string}
  b = b_0 b_1 \dots b_{k-1}\,.
\end{equation}
The number $k$ is the \emph{length} of $b$ and written as $\abs{b}$.
For the empty string we write $\epsilon$.

\emph{Substrings} of a string or a configuration are written in the
form
\begin{equation}
  \label{eq:substring}
  a_{j:k} = a_j a_{j+1} \dots a_{j+k-1}\,.
\end{equation}

Finally, there is the \emph{particle content} of a cell, string, or
configuration. For a single cell $\alpha$, it is written $\# \alpha$,
so that we have $\# \pi_i = i$. For a string, it is
\begin{equation}
  \label{eq:string-content}
  \# a = \sum_i \# a_i \,.
\end{equation}
The content of a configuration only exists if all except of a finite
number of its cells are empty; then it is defined as
in~(\ref{eq:string-content}).

For strings and cells we will also use the \emph{complementary content}
\begin{equation}
  \label{eq:complementary-content}
  \#^c a = C \abs{a} - \#a\,.
\end{equation}
It is the number of particles that the cells in $a$ can accept
particles until all of them are full.

\subsection{Transition Rules}

That what turns a cellular automaton into a dynamical system is the
existence of a \emph{global transition rule}
$\hat{\phi} \colon \Sigma^\Z \to \Sigma$. The sequence
\begin{equation}
  c, \hat{\phi}(c), \hat{\phi}^2(c), \dots, \hat{\phi}^t(c), \dots
\end{equation}
is the \emph{evolution} (or \emph{orbit}) of an initial configuration
$c$; the configuration $\hat{\phi}^t(c)$ is the state of the cellular
automaton at \emph{time step} (or \emph{generation}) $t$.

Among the discrete dynamical systems with doubly infinite cell
sequences as configurations, cellular automata are distinguished by
the existence of two integers $r_1$, $r_2 \geq 0$ and a \emph{local
  transition rule} $\phi \colon \Sigma^{r_1+r_2+1} \to \Sigma$ so that for all
$i \in \Z$,
\begin{equation}
  \label{eq:local-global-rule}
  \hat{\phi}(c)_i = \phi(c_{i-r_1:r_1+r_2+1}) \,.
\end{equation}
This means that the next state of the cell in position $i$ only
depends on a neighbourhood that reaches $r_1$ cells to the left of
this cell and $r_2$ cells to the right. The numbers $r_1$ and $r_2$
are called the \emph{left} and \emph{right radius} of the transition
rule.

Now consider a second cellular automaton $\hat{\phi}'$ with the same
local transition function $\phi$ but left and right radii
$r'_1 = r_1 + k$ and $r'_2 = r_2 - k$ for some integer $k$. The two
automata show almost the same behaviour except that under
$\hat{\phi}'$, the states of the cells in the next generation are
shifted $k$ cells to the right (for positive $k$) compared to the
behaviour under $\phi$. We call two such automata \emph{equivalent}.

\begin{figure}[t]
   \begin{evolution}
     \evol{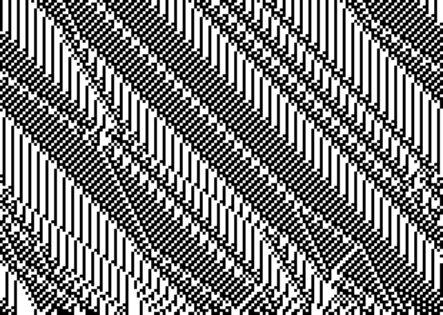} & \evol{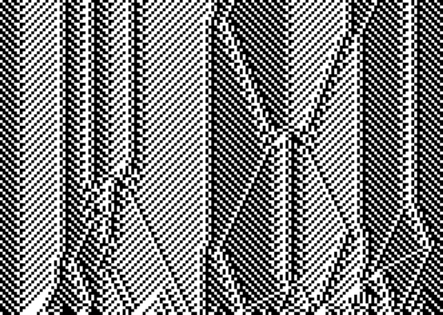} \\
     \evol{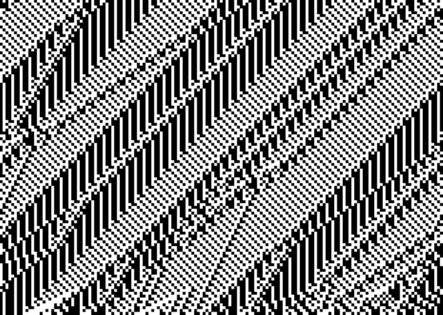} & \evol{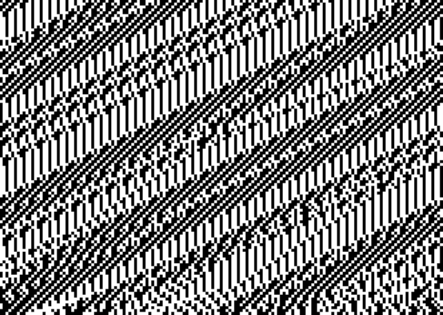}
   \end{evolution}
   \caption{Equivalent cellular automata and the ``stick mantis'' rule. \\
     These images show equivalent cellular automata that differ only
     in their horizontal shift. In an image of a cellular automaton
     evolution, the structures that do
     not move especially stand out, so that each of the images highlights
     different aspects of essentially the same flow rule. But the top
     right automaton, with its mantis-like structures, is certainly the
     most natural version.
     \\
     \emph{Parameters:} $C = 1$, $\delta = (0.5, 0.5)$; top left:
     $f = m(1,100;0) \join m(1,010;1)$, top right:
     $f = m(11,00;1) \join m(10,10;1)$, bottom left:
     $f = m(110,0;1) \join m(101,0;2)$, bottom right:
     $f = m(1100,1) \join m(1010,2)$.}
  \label{fig:equivalent}
\end{figure}
Among equivalent cellular automata, those with a one-sided
neighbourhood are the easiest to treat theoretically, for example in
proofs. We will use here the ``left rules'', that are those with
$r_2 = 0$, whenever we speak about cellular automata in a more
abstract context. Concrete automata may however look vastly different
even if they are equivalent (see Figure~\ref{fig:equivalent}), and so
we need also to consider arbitrary values for $r_1$ and $r_2$.

\subsection{Number Conservation}

We now define a global transition rule $\hat{\phi}$ as
\emph{number-conserving} if
\begin{equation}
  \label{eq:number-conserving}
  \# \hat{\phi}(c) = \# c
\end{equation}
for all configurations $c \in \Sigma^\Z$ for which $\# c$ exists. The
local transition rule $\phi$ that belongs to $\hat{\phi}$ is then also
number-conserving.


The following theorem -- which incorporates several definitions --
recapitulates all the facts about number conservation in
\cite{Redeker2022} that will be used in this chapter.

\begin{theorem}[Number conservation and flows]
  \label{thm:number-conservation}
  \begin{enumerate}
  \item Let $\phi \colon \Sigma^{r_1+r_2+1} \to \Sigma$ be the local transition
    function of a number-conserving cellular automaton and let $C$ be
    its capacity. We set $\ell = r_1 + r_2$ and call it the \emph{flow
      length} of $\phi$ -- it will occur over and over on the following
    pages.

    Then there exists a \emph{flow function}
    $f \colon \Sigma^\ell \to \Z$ for $\phi$ such that for all $a \in \Sigma^{\ell + 1}$,
    \begin{equation}
      \label{eq:flow-definition}
      \# \phi(a) = \# a_{r_1} + f(a_{0:\ell}) - f(a_{1:\ell})\,.
    \end{equation}
    For our choice of $\Sigma$, the function $f$ therefore determines
    $\phi$ completely. \cite[Thm.~3.1]{Redeker2022}

  \item Let $a \in \Sigma^{r_1}$ and $b \in \Sigma^{r_2}$. Then the value of
    $f(a b)$ is bounded by the inequalities
    \begin{subequations}
      \begin{equation}
        \label{eq:flow-bounds}
        - \# b \leq f(a b) \leq \# a\,.
      \end{equation}
      For $r_2 = 0$, this becomes
      \begin{equation}
        \label{eq:flow-bounds-left}
        0 \leq f(a) \leq \# a \,.
      \end{equation}
    \end{subequations}

  \item We will write $\mathcal{F}$ or $\mathcal{F}(r_1, r_2, C)$ for the set of all flows
    with given $r_1$, $r_2$ and $C$. This set is a distributive
    lattice \cite[Chap.~4]{Davey2002} with its operations given by
    \begin{subequations}
      \begin{align}
        \label{eq:lattice-meet}
        (f \meet g)(a) &= \min\{ f(a), g(a)\}, \\
        \label{eq:lattice-join}
        (f \join g)(a) &= \max\{ f(a), g(a)\}
      \end{align}
      \label{eq:lattice-operations}
    \end{subequations}
    for all $f$, $g \in \mathcal{F}$ and $a \in \Sigma^\ell$. The
    order relation on this lattice is therefore
    \begin{equation}
      \label{eq:lattice-order}
      f \leq g
      \qquad\text{iff}\qquad
      \forall a \in \Sigma^\ell \colon f(a) \leq g(a)\,.
    \end{equation}

    The flow lattices $\mathcal{F}(r_1, r_2, C)$ with the same value
    of $\ell$ are isomorphic, and we can therefore speak of the
    lattice $\mathcal{F}(\ell, C)$. \cite[Thm.~6.1]{Redeker2022}

  \item\label{item:minimal-elements} The minimal elements of
    $\mathcal{F}$ according to the lattice order~\eqref{eq:lattice-order} are
    the flows $m(a, b; k)$ with $a \in \Sigma^{r_1}$,
    $b \in \Sigma^{r_2}$, $-\#b \leq k \leq \# a$ and
    \begin{equation}
      \label{eq:minimal-elements}
      m(a, b; k) =
      \bigmeet \set{ f \in \mathcal{F} \colon f(a b) \geq k}\,.
    \end{equation}
    In case of $r_2 = 0$, we write $m(a, k)$ for $m(a, \epsilon; k)$.
    The set of all these \emph{minimal flows} is $\Fmin$.

    As it is true for all distributive lattices, all flows can be
    written as a maximum of a set of minimal flows.
    \cite[Sect.~6.3]{Redeker2022}
    
  \item For the left neighbourhood, there is an explicit form for
    $m(a, k)$. Namely, if $a$, $b \in \Sigma^\ell$, then the function value
    $m(a, k)(b)$ is the smallest number so that
    \begin{subequations}
      \begin{align}
        \label{eq:minflow-def-0}
        m(a, k)(b) &\geq 0, \\
        \label{eq:minflow-def-min}
        m(a, k)(b) &\geq
                     k - \min\{\abs{u}, \abs{u'}\} C - \# w - \#^c w'
      \end{align}
      \label{eq:minflow-def}
    \end{subequations}
    for all factorisations $a = u v w$, $b = u' v w'$.
    \cite[Thm.~6.2]{Redeker2022}

    For an arbitrary neighbourhood, we have
    \begin{equation}
      \label{eq:minflow-bothside}
      m(a_1, a_2; k)(b) = m(a_1 a_2, k) - \# a_2
    \end{equation}
    where $a_1 \in \Sigma^{r_1}$ and $a_2 \in \Sigma^{r_2}$.
    \cite[Sect.~7.2]{Redeker2022}
  \end{enumerate}
\end{theorem}

Equation~\eqref{eq:minflow-bothside} is also the reason why all flow
lattices with the same values for $\ell$ and $C$ are isomorphic.

A flow function measures the particle movement across a boundary in a
configuration. This means: Declare a boundary between two cells an a
configuration $c$, dividing it into a left side $c_L$ and a right side
$c_R$. Let the string $a$ consist of the $r_1$ cells left of the
boundary and the string $b$ of the $r_2$ cells right of it. Then
$f(a b)$ is the number of particles that cross the boundary to the
right during the transition from $c$ to $\hat{\phi}(c)$. A negative
value means that there is a movement to the left. (Or in a more
prosaic way: The particle content of $c_L$ decreases by $f(a b)$ while
that of $c_R$ increases by the same amount.)

This also explains the bounds in~\eqref{eq:flow-bounds}: There is no
guarantee that more than $\# a$ particles are in $c_L$, or $\# b$
particles in $c_R$, which puts limits on the number of particles that
can cross the boundary.

\subsection{Example: The Shift Rules}
\label{sec:shift-rules}

The shift rules are certainly the only number-conserving rules that
everyone knows. They therefore serve as a good example how the theory
of flows works in a concrete case.

A \emph{shift rule} is one of the transition rules $\hat{\sigma}^k$ for which
\begin{equation}
  \label{eq:shift-rule}
  \hat{\sigma}(c)_i = c_{i-k}
\end{equation}
for all $c \in \Sigma^\Z$ and $i \in \Z$. It is the rule that moves all
particles $k$ cells to the right.

This translates easily into a flow function. If $k \geq 0$, then the
function $f \colon \Sigma^k \to \Z$ with $f(a) = \# a$ is the flow function
that belongs to $\hat{\sigma}^k$. It is a ``left'' flow function, analogous
to the left transition rules, with $r_1 = k$ and $r_2 = 0$. One can
then check, via~\eqref{eq:flow-definition}, that the function
$\sigma^k \colon \Sigma^{k+1} \to \Sigma$ with $\sigma^k(a) = a_0$ is the transitions
function that belongs to $f$. The case for $k \leq 0$ is symmetric.

Later, in Theorem~\ref{thm:shift-flows}, we will see that the flow
function that belongs to $\hat{\sigma}^k$ is also a minimal flow.

\section{Algorithmic Aspects}

\subsection{Computation of the Transition Rule}

In order to experiment with number-conserving cellular automata, the
definitions of Theorem~\ref{thm:number-conservation} must be
translated into a program. For my own experiments, I have used the
following data structures:
\begin{itemize}
\item A class \class{MinFlow} which represents a minimal flow.

  An object of this class stores the values $a \in \Sigma^\ell$ and
  $k \in \Z$ and can compute $m(a, k)(b)$ for all $b \in \Sigma^\ell$.
  
\item A class \class{FlowSum} which stores a non-empty set $S$ of
  \class{MinFlow} objects. It can compute the flow
  $F(b) = \bigjoin_{f\in S} f(b)$ for all $b \in \Sigma^\ell$.

  For efficiency reasons, the set $S$ must also be an \emph{antichain}
  \cite[Sect.~1.3]{Davey2002}, which means that there are no $f$,
  $g \in S$ with $f < g$. It is clear that all flows can be realised
  with an $S$ that is an antichain.
  
\item A class \class{Rule} which stores a \class{FlowSum} object $F$
  and can compute
  \begin{equation}
    \label{eq:rule-value}
    \# \phi(u) = \# u_\ell + F(u_{0:\ell}) - F(u_{1:\ell})
  \end{equation}
  and therefore $\phi(u)$ for all $u \in \Sigma^{\ell + 1}$.
\end{itemize}

It would look as if these classes could only compute $\phi$ for a
left-side neighbourhood. But the two-sided flow functions
$m(a, a'; k)$ are constructed in such a way that $m(a, a'; k)$ has the
same local transition rule $\phi$ as $m(a a', k)$, and this is also true
for the maximum of a set of such functions.

Therefore, to find the value of $\phi(b)$ for
$F = \bigjoin_{i=1}^n m(a_i, a'_i; k)$, it is enough to compute it from
$\bigjoin_{i=1}^n m(a_i a'_i, k)$ with the data structures just shown.

For efficiency, it is better not to compute the values of the flow
function and the transition rule each time they are needed. A
\class{MinFlow} object should instead store the map $b \mapsto m(a, k)(b)$
in memory, and a \class{SumFlow} object should store $u \mapsto \phi(u)$. If
$\ell$ or $C$ are large, these maps might become quite big, however. Then
a ``lazy'' approach is useful, where these maps are initially empty
and are only populated when a value is computed.

There remains the computation of the values of $m(a, k)$.
Algorithm~\ref{alg:minimal-flow} shows a possibility. It uses the
notation of Section~\ref{sec:cell-states}.

\begin{algorithm}[t]
  \SetKwData{val}{val}
  \SetKwData{newval}{new\_val}
  \SetKwData{ulen}{u\_len}
  \SetKwData{vlen}{v\_len}
  \SetKwData{wlen}{w\_len}
  \SetKwData{uulen}{u'\_len}
  \SetKwData{wwlen}{w'\_len}

  \KwIn{$a \in \Sigma^\ell$, $k \in \Z$, $b \in \Sigma^\ell$}
  \KwOut{$m(a, k)(b)$}
  \BlankLine

  $\val \gets 0$\;
  \For{$\ulen \in \{0, \dots, \ell\}$}{
    \For{$\vlen \in \{0, \dots, \ell - \ulen\}$}{
      $\wlen \gets \ell - \ulen - \vlen$\;
      $v \gets a_{\ulen:\vlen}$\;
      $w \gets a_{\ulen+\vlen:\wlen}$\;
      \For{$\uulen \in \{0, \dots, \ell - \vlen\}$}{
        \If{$b_{\uulen:\vlen} = v$}{
          $\wwlen \gets \ell - \uulen - \vlen$\;
          $w' \gets b_{\uulen+\vlen:\wwlen}$\;
          $\newval \gets k - \min(\ulen, \uulen) \times C - \# w - \#^c w'$\;
          $\val \gets \max(\val, \newval)$\;
        }
      }
    }
  }
  \Return{\val}\;
  \caption{Values of $m(a, k)$.}
  \label{alg:minimal-flow}
\end{algorithm}

\subsection{Exploring the Rule Space}
\label{sec:exploring-rule-space}

After implementing the transition rule algorithm, the next thing to do
is finding interesting rules. This can be done with a variant of the
mutation algorithm in Andrew Wuensche's \texttt{ddlab} program
\cite[Sect.~32.5.4]{Wuensche2018}.

Wuensche's algorithm modifies the local transition rule $\phi$ by
changing its output for one randomly chosen neighbourhood. One
explores the rule space by starting with an arbitrary rule and
mutating it repeatedly until a rule with an interesting behaviour
appears. In \texttt{ddlab}, it is also possible to step back to
earlier rules if one is in a boring region of the rule space.

But this algorithm does not preserve number conservation. We need a
new one that does this. It is an algorithm with two kinds of mutations
and a method to choose one of them at random. The algorithm operates
directly on \class{FlowSum} objects; it modifies the set
$S \subseteq \Fmin$ that represents flows of the form
$F = \bigjoin_{f \in S} f$.

We have then the following mutation algorithms:
\begin{description}
\item[\textbf{Upward Mutation:}] Choose a random $a \in \Sigma^\ell$ and set
  $k = F(a) + 1$. Replace $S$ with $S \cup \{ m(a, k) \}$, then
  simplify\footnote{Simplification will be described below in
    Section~\ref{sec:simplification-comparison}.} $S$.

  This is only possible if $k \leq \# a$, since otherwise $m(a, k)$ is
  not a flow. In this case, the mutation is not successful.

\item[\textbf{Downward Mutation:}] Choose a random element
  $m(a, k) \in S$ and replace $S$ with
  $S \setminus \{ m(a, k) \} \cup \{ m(a, k - 1) \}$, then simplify $S$.

  This is only possible if $k > 0$, since otherwise $m(a, k - 1)$ is
  not a flow. In this case, the mutation is not successful.

\item[\textbf{Mutation:}] Try repeatedly and at random either an
  upward or a downward mutation until one of them succeeds.

  The algorithm terminates because for every set $S$, either an upward
  or a downward mutation is always possible.
\end{description}

This mutation algorithm works well unless the values of $\ell$ or $C$ are
large. Then the search stays for a long time in boring regions of the
rule space where the evolution of the cellular automata looks chaotic
with no discernible features.

\begin{figure}[t]
   \begin{evolution}
     \evol{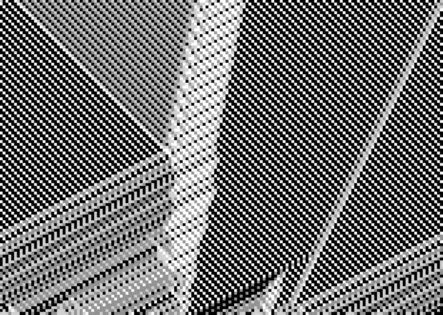} & \evol{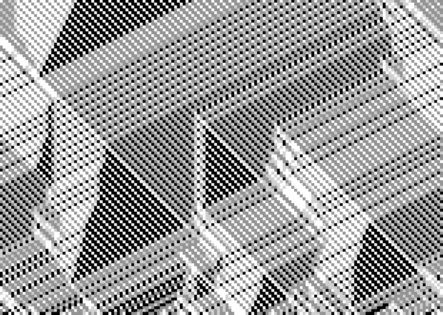} \\
     \evol{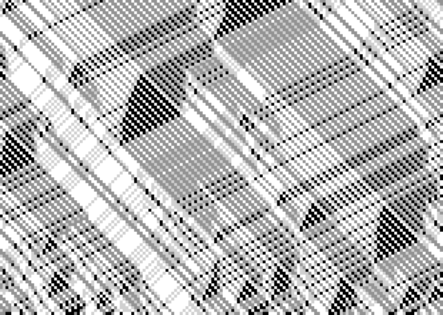} & \evol{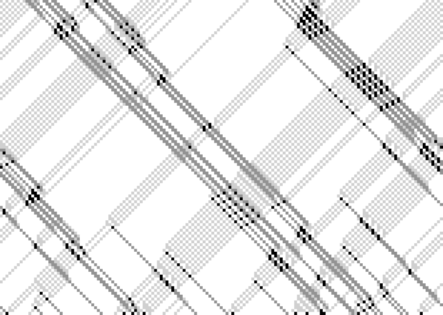}
   \end{evolution}
   \caption{Different densities, with a ``crystalline'' rule. \\
     One can see this as the transition from a ``solid'' state to a
     ``gaseous'' state, with two images of the ``liquid'' state in
     between. Could there be a meaningful thermodynamics for
     number-conserving cellular automata?
     \\
     \emph{Parameters:} $C = 3$,
     $f = m(01,0;1) \join m(03,0;2) \join m(11,0;1) \join m(13,0;2)
     \join m(30,0;1) \join m(30,1;1) \join m(31,0;1) \join m(33,3;2)$.
     Top left: $\delta = (0.4, 0.32, 0.32, 0.32)$, top right:
     $\delta = (0.25, 0.25, 0.25, 0.25)$, bottom left:
     $\delta = (0.49, 0.17, 0.17, 0.17)$, bottom right:
     $\delta = (0.82, 0.06, 0.06, 0.06)$.}
  \label{fig:density}
\end{figure}
When searching, it is good to vary the initial particle density from
time to time: Number-conserving cellular automata have often strongly
differing behaviour at different particle densities
(Figure~\ref{fig:density}). One should also switch between equivalent
transition rules: Patterns that appear in the evolution of one variant
may not be visible with another variant of the same rule.

Finally it is useful if besides the general mutation command, one has
explicit commands for upward and downward mutations. They help to move
the search out of uninteresting regions.

All number-conserving rules in the illustrations were found with these
methods.

\subsection{Simplification and Comparison}
\label{sec:simplification-comparison}

What remains to describe is the simplification process. The
\emph{simplification} of a set $S \subseteq \mathcal{F}$ is the smallest set $S' \subseteq S$
such that $\bigjoin_{f \in S'} f = \bigjoin_{f \in S} f$. It is easy to
see that
\begin{equation}
  \label{eq:simplification}
  S' = \set{ f \in S \colon \not\exists g \epsilon S \colon f < g}
\end{equation}
and that $S'$ therefore is an antichain, as required in the definition
of \class{FlowSum}. The only problem, from an algorithmic viewpoint,
is the comparison $f < g$.

If we were to compare $f$ and $g$ directly according
to~\eqref{eq:lattice-order}, we would have to compare $f(a)$ and
$g(a)$ for all $a \in \Sigma^\ell$, that are $(C + 1)^\ell$ comparisons. But in
case of the mutation algorithm, we always compare minimal flows. And
then there is a faster method, as the following theorem shows.

\begin{theorem}[Comparison]
  \label{thm:minimal-compare}
  Let $m(a, j)$ and $m(b, k)$ be two minimal flows. Then
  \begin{equation}
    \label{eq:minimal-compare}
    m(a, j) \leq m(b, k)
    \qquad\text{iff}\qquad
    m(a, j)(a) \leq m(b, k)(a)\,.
  \end{equation}
\end{theorem}
\begin{proof}
  We set $f_a = m(a, j)$ and $f_b = m(b, k)$. Then the theorem states
  that $f_a \leq f_b$ if and only if $f_a(a) \leq f_b(a)$.

  First we assume that $f_a(a) \leq f_b(a)$.
  From~\eqref{eq:minimal-elements} now follows that
  \begin{align}
    f_a = m(a, \epsilon; j)
    &= \bigmeet\set{g \in \mathcal{F} \colon g(a) \geq j} \notag \\
    &= \bigmeet\set{g \in \mathcal{F} \colon f_a(a) \leq g(a)}
  \end{align}
  and that if $f_a(a) \leq g(a)$, then $f_a \leq g$. Replacing $g$ with
  $f_b$, we then can conclude that from $f_a(a) \leq f_b(a)$ follows
  $f_a \leq f_b$. This proves one direction of the theorem.

  For the other direction, we notice that there are three
  possibilities for the relation between $f_a$ and $f_b$: We might
  have $f_a(a) \leq f_b(a)$ and therefore $f_a \leq f_b$, or
  $f_b(b) \leq f_a(b)$ and therefore $f_b \leq f_a$. If none of this is
  true, we must have $f_a(a) > f_b(a)$ and $f_b(b) > f_a(b)$; then
  $f_a$ and $f_b$ are incomparable. This shows that $f_a \leq f_b$ can
  only happen if $f_a(a) \leq f_b(a)$.
\end{proof}

Therefore, the simplification algorithm has the form: For all
$m(a, j)$ and $m(b, k) \in S$ with $j \leq m(b, k)(a)$, remove
$m(a, j)$ from $S$.

\section{The Lattice of Flow Functions}
\label{sec:lattice-of-flow-functions}

Due to the complexity of the formula~\eqref{eq:minflow-def-min} for
$m(a, k)$, it is not that simple to get information about the lattice
$\mathcal{F}(\ell, C)$, but some results exist.

\begin{theorem}[Bounds]
  \label{thm:bounds}
  We have
  \begin{equation}
    \label{eq:bounds}
    k - \#^c b \leq m(a, k)(b) \leq k
  \end{equation}
  for all $m(a, k) \in \mathcal{F}(\ell, C)$ and $b \in \Sigma^\ell$.
\end{theorem}
\begin{proof}
  By~\eqref{eq:minflow-def}, $m(a, j)(b)$ is in essence the minimum of
  terms of the form $k - \min\{\abs{u}, \abs{u'}\} - \# w - \#^c w'$,
  where only non-negative terms are subtracted from $k$. This proves
  the upper bound.

  For the lower bound, we set $u=a$, $v = w = u' = \epsilon$ and $w' = b$
  in~\eqref{eq:minflow-def-min} and get
  \begin{align}
    m(a, k)(b)
    &\geq k - \min\{\abs{u}, \abs{u'}\} - \# w - \#^c w' \notag \\
    &= k - 0 - 0 - \#^c b
  \end{align}
  as result.
\end{proof}
We could in a similar way, by setting $u = v = w' = \epsilon$, $w = a$ and
$u' = b$, derive the inequality $m(a, k)(b) \geq k - \# a$, but it is
useless: It brings no improvement over the already existing bound
$m(a, k)(b) \geq 0$ because we know from
Theorem~\ref{thm:number-conservation},
item~\ref{item:minimal-elements}, that $k \leq \# a$.

But at least Theorem~\ref{thm:bounds} allows us to evaluate $m(a, k)$
completely in a very simple case. The following theorem formalises a
statement that already occured in \cite{Redeker2022}.

\begin{theorem}[Flows of  shift rules]
  \label{thm:shift-flows}
  The function $m(\pi_C^\ell, \ell C)$ is the flow for the shift rule $\hat{\sigma}^\ell$.
\end{theorem}
Note that $\pi_C^\ell$ is the neighbourhood that consists of $\ell$ completely
filled cells and that $\ell C$ is its particle content.

\begin{proof}
  As we have seen in Section~\ref{sec:shift-rules}, the flow function
  for $\hat{\sigma}^\ell$ maps any $b \in \Sigma^\ell$ to $\# b$.

  To prove this, we derive from~\eqref{eq:bounds} the inequality
  \begin{equation}
    \ell C - \#^c b \leq m(\pi_C^\ell, \ell C)(b) \leq \#b,
  \end{equation}
  which is equivalent to  $m(\pi_C^\ell, \ell C)(b) = \# b$, as required.
\end{proof}

This then means that $m(\pi_C^\ell, \ell C)$ is the maximal element in
$\mathcal{F}(\ell, C)$: It is the function that for every argument $b$ takes the
highest possible value, $\# b$. The minimal element of $\mathcal{F}(\ell, C)$ is,
trivially, the ``non-flow'' function $m(\pi_0^C, 0)$, which for every
argument takes the value $0$.

What else can be said about the lattices $\mathcal{F}(\ell, C)$? We have already
seen in \cite[Sect.~8]{Redeker2022} that $\mathcal{F}(\ell, C)$ is not a linear
order and contains incomparable elements, at least for large enough
$\ell$. And there is also an isomorphic embedding from
$\mathcal{F}(\ell, C)$ to every lattice $\mathcal{F}(\ell, n C)$. The idea behind it is to
replace each particle in the original with $n$ particle in the
transformed automaton and to adapt the transition rules appropriately.

\begin{figure}[t]
  \begin{evolution}
    \evol{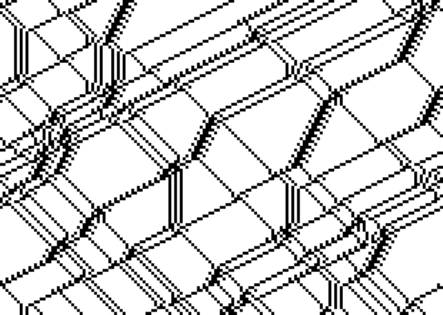} & \evol{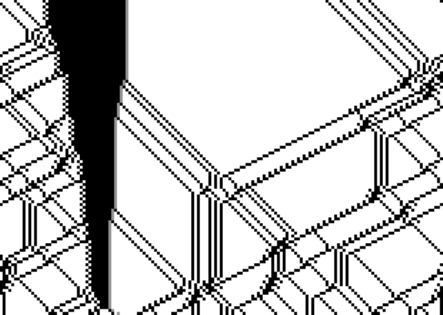} \\
    \evol{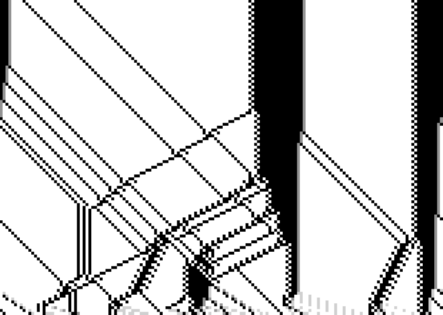} & \evol{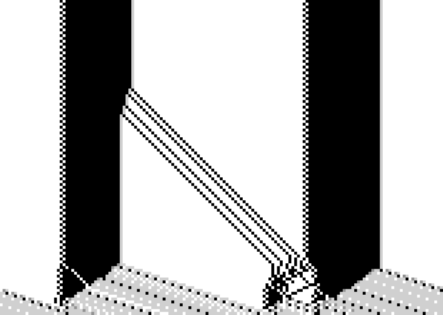}
  \end{evolution}
  \caption{Theorem~\ref{thm:embedding} and the ``staircase'' rule. \\
    \emph{Top left:} A cellular automaton with capacity 1 and an
    interesting particle structure.
    \emph{Parameters:} $C = 1$, $f = m(101,\br 001;\br 1) \join
    m(010,110;1) \join m(011,000;2) \join  m(011,101;1) \join
    m(110,011;1)$, $\delta = (0.8, 0.2)$. \\
    \emph{Top right:} Simulation of this automaton with capacity 3,
    very slightly disturbed. The growing black wedges have at the
    right side a cell with two particles, a defect that does not go
    away. \emph{Parameters:} $C = 3$,
    $f = m(303,003;3) \join m(030,330;3) \join m(033,000;6) \join
    m(033,303;3) \join m(330,033;3)$, $\delta = (0.79, 0, 0.02, 0.19)$. \\
    \emph{Bottom left:} If ones adds more disturbances, the time to
    organise the particles into 3-blocks grows.
    $\delta = (0.64, 0.2, 0.01, 0.14)$ \\
    \emph{Bottom right:} With a greater number of $\pi_1$ cells, the
    unorganised phase persists longer. For many random choices of the
    initial configuration, no interaction occurs.
    $\delta = (0.16, 0.8, 0, 0.03)$.}
  \label{fig:simulation}
\end{figure}
\begin{theorem}[Embedding]
  \label{thm:embedding}
  For every $n > 0$ exists an injective lattice homomorphism
  $e \colon \mathcal{F}(\ell, C) \to \mathcal{F}(\ell, n C)$ and a \emph{state embedding
    function} $\eta \colon \Sigma_C^* \to \Sigma_{nC}^*$ such that
  \begin{equation}
    \label{eq:embedding}
    e(f)(\eta(b)) = n f(b)
  \end{equation}
  for all $f \in \mathcal{F}(\ell, C)$ and $b \in \Sigma_C^*$.
\end{theorem}
\begin{proof}
  We set $\eta(\pi_{i_1} \dots \pi_{i_k}) = \pi_{n i_1} \dots \pi_{n i_k}$ for
  all finite strings $\pi_{i_1} \dots \pi_{i_k}$. The map $e$ is defined
  on minimal flows via
  \begin{equation}
    \label{eq:embedding-minimal-flows}
    e(m(a, k)) = m(\eta(a), n k)
  \end{equation}
  and then extended to arbitrary flows.

  To show that this construction is valid, we evaluate
  conditions~\eqref{eq:minflow-def} for $e(m(a, k))(\eta(b))$ and
  show that they always request a value that is $n$ times as much as
  the value requested by the same conditions for $m(a, k)(b)$.

  Condition~\eqref{eq:minflow-def-0} is trivial.
  From~\eqref{eq:minflow-def-min}, we get
  \begin{align}
    \label{eq:embedding-proof}
    e(m(a, k))(\eta(b))
    &= m(\eta(a), n k)(\eta(b)) \notag \\
    &\geq n k - \min\{\abs{\eta(u)}, \abs{\eta(u')}\} n C
      - \# \eta(w) - \#^c \eta(w') \notag \\
    &= n(k - \min\{\abs{u}, \abs{u'}\} - \# w - \#^c w'),
  \end{align}
  which together is enough to prove that
  $e(m(a, k))(\eta(b)) = n m(a, k)(b)$.

  This means that $e$ preserves the ordering between the minimal
  flows and that it indeed can be extended to an injective lattice
  homomorphism on $\mathcal{F}(\ell, C)$.
\end{proof}

This homomorphism makes it possible to simulate every
number-conserving with capacity $C$ by one with capacity $n C$. In the
simulation, every cell in the simulating automaton contains a number
of particles that is divisible by $n$. But what happens when we start
the ``simulation'' from an initial configuration in which the number
of particles per cell is unrestricted?

A similar question arose in \cite{Boccara1991a}. There, the simulating
cellular automaton operated on configurations in which each cell of
the original automaton was replaced with a block of $n$ adjacent equal
cells. And, for totalistic rules, the cells of an unrestricted initial
configuration evolved into blocks of $n$ cells, with some defects in
between.

This lets one expect that in our case, most of the cells will soon
contain a multiple of $n$ particles, again with some defects.
Figure~\ref{fig:simulation} supports this.

Finally, there is a strange relation -- a kind of reciprocity law --
that directly follows from the definition of $m(a, k)$. To express it,
we define the \emph{complement} of a string
$a =\pi_{i_1} \dots \pi_{i_k} $ as
$\cpl a = \pi_{C - i_1} \dots \pi_{C - i_k}$. Then we have the following
theorem:
\begin{theorem}[Reciprocity]
  Let $m(a, k) \in \mathcal{F}(\ell, C)$ and $b \in \Sigma^\ell$. Then
  \begin{equation}
    \label{eq:reciprocity}
    m(a, k)(b) = m(\cpl b, k)(\cpl a)\,.
  \end{equation}
\end{theorem}
\begin{proof}
  When switching to complements in~\eqref{eq:minflow-def-min}, the
  operators $\#$ and $\#^c$ effectively change their roles, and this
  makes the switch from $a$ and $b$ to $\cpl b$ and $\cpl a$ possible.
\end{proof}

The effect of this relation is especially interesting if we apply it
to an arbitrary flow $F = \bigjoin_{i=1}^n m(a_i, k)$, because then we
have
\begin{equation}
  \label{eq:reciprocity-example}
  F(b) = \bigjoin_{i=1}^n m(\cpl b, k)(\cpl{a_i})\,.
\end{equation}
This equation tells us that the evaluation of an arbitrary flow
function $F$ at a neighbourhood $b$ can also achieved by evaluation of
a single minimal flow, $m(\cpl b, k)$, but at several different
neighbourhoods. What follows from this relation is however unclear.

\section{What Next to Do}

All investigations in this chapter centred on the minimal flows
$m(a, k)$. We have seen that they make the work with number-conserving
rules easier, but have much intrinsic complexity themselves.

The minimal flows provide names for arbitrary flow functions -- which
was actually the primary reason for defining them -- and they make
calculations with flows easier, as in
Theorem~\ref{thm:minimal-compare}. But they are computed via
formula~\eqref{eq:minflow-def-min} as a maximum over a large number of
terms: This makes Algorithm~\ref{alg:minimal-flow} complicated and the
theorems of Section~\ref{sec:lattice-of-flow-functions} somewhat
unsatisfying.

A better understanding of the minimal flows is therefore the most
important goal in the further development of the theory of
one-dimensional number-conserving cellular automata.

\bibliography{../references}
\end{document}